\documentclass{xarticle}
\usepackage{graphicx}
\usepackage{overpic}
\usepackage{stmaryrd}
\usepackage{trimclip}
\usepackage{mathtools}
\usepackage{xspace}
\usepackage[numbers,sort]{natbib}

\newcommand{\eemph}[1]{\textbf{\textit{#1}}}
\newcommand{\Z}{\mathbb{Z}}
\newcommand{\R}{\mathbb{R}}

\newcommand{\bmat}[1]{\begin{bmatrix}#1\end{bmatrix}}
\def\one{\mathbf{1}}
\def\tp{\mathsf{T}}
\let\mathopfont=\mathrm
\newcommand{\Null}{\mathop{\mathopfont{null}}}
\newcommand{\Span}{\mathop{\mathopfont{span}}}
\newcommand{\rank}{\mathop{\mathopfont{rank}}}
\newcommand{\range}{\mathop{\mathopfont{range}}}
\let\bl\bigl
\let\br\bigr
\newbox\vcbox
\def\vcent#1{\setbox\vcbox\hbox{#1}\raise -0.5\ht\vcbox\hbox{#1}}
\newcommand{\ie}{{\it i.e.}}
\newcommand{\floor}[1]{\lfloor#1\rfloor}

\newcommand{\bt}{bittide\xspace}

\def\itoj{{i \shortto j}}
\def\jtoi{{j \shortto i}}

\def\wu{\omega^\text{u}}
\def\D{\mathbb{D}}
\def\A{\mathbb{A}}

\DeclarePairedDelimiter\abs{\lvert}{\rvert}

\def\din{d^\text{\hskip0.1em in}}
\def\thetass{\theta^\text{ss}}
\def\wss{\omega^\text{ss}}

\def\bwss{{\bar\omega}^\text{ss}}
\def\zetass{\zeta^\text{ss}}
\def\bzetass{{\bar\zeta}^\text{ss}}
\def\betass{\beta^\text{ss}}

\def\betaoff{\beta^\text{off}}
\def\kp{k}
\def\hash{\texttt{\#}}

\makeatletter
\DeclareRobustCommand{\shortto}{\mathrel{\mathpalette\short@to\relax}}
\newcommand{\short@to}[2]{\mkern2mu
  \clipbox{{.5\width} 0 0 0}{$\m@th#1\vphantom{+}{\shortrightarrow}$}}
\makeatother

\begin{document}

\title{Modeling Buffer Occupancy in \bt Systems}

\author{%
  Sanjay Lall\footnotesymbol{1}
  \and  Tammo Spalink\footnotesymbol{2}}

\note{Preprint}

\maketitle

\makefootnote{1}{S. Lall is with the Department of Electrical
  Engineering at Stanford University, Stanford, CA 94305, USA, and 
  at Google DeepMind.
  \texttt{lall@stanford.edu}\medskip}

\makefootnote{2}{Tammo Spalink is at Google DeepMind.}

\begin{abstract}

The \bt mechanism enables logically synchronous computation 
across distributed systems by leveraging the continuous frame 
transmission inherent to wired networks such as Ethernet.  
Instead of relying on a global clock, \bt uses a decentralized 
control system to adjust local clock frequencies, ensuring all 
nodes operate with a consistent notion of time by utilizing 
elastic buffers at each node to absorb frequency variations.  
This paper presents an analysis of the steady-state occupancy 
of these elastic buffers, a critical factor influencing system 
latency.  Using a fluid model of the \bt system, we 
prove that buffer occupancy converges and derive an explicit 
formula for the steady-state value in terms of system parameters, 
including network topology, physical latencies, and controller 
gains. This analysis provides valuable insights for optimizing 
buffer sizes and minimizing latency in \bt-based distributed 
systems.

\end{abstract}

\section{Introduction}

Digital circuits achieve deterministic and predictable behavior through 
synchronous logic, where state transitions are precisely aligned to a 
single clock signal. This clock-driven coordination is fundamental to 
the efficient and reliable operation of computations within a single 
integrated circuit.  However, extending this synchronous paradigm to 
the scale of an entire datacenter, encompassing a large network of 
interconnected machines, presents significant challenges. Traditional 
approaches rely on synchronizing individual clocks to a global time 
reference, introducing overheads and scalability bottlenecks.  The 
\bt mechanism offers a novel solution by establishing \emph{logical 
synchrony}, effectively extending the principles of clock-driven 
coordination from the granularity of individual integrated circuits 
to the scale of datacenters.

The \bt system achieves this by leveraging the continuous frame transmission 
inherent to modern wired networks like Ethernet.  Instead of aligning 
clocks to a single time source, \bt ties local clock advancements 
directly to these continuous frame transmissions. This creates a 
decentralized synchronization scheme with minimal overhead, as it 
requires no explicit exchange of timing information~\cite{ls}. 
This decentralized approach, in contrast to traditional clock 
synchronization protocols like PTP~\cite{ptp}, allows for 
significantly improved scalability and robustness. As in 
synchronous digital circuits, events within a single node are 
ordered by the local clock, while events across different nodes are 
ordered by causality, determined by the flow of data frames.

At the core of the \bt mechanism is a distributed feedback control 
system. This system monitors communication rates between adjacent nodes, 
adjusting local clock frequencies to maintain logical synchrony.  
Elastic buffers at each node absorb frequency variations arising from 
imperfect physical clocks.  A critical function of the control system 
is to prevent these buffers from overflowing or underflowing, 
guaranteeing data integrity and system stability.

While prior research~\cite{spalink_2006,bms,ls,res,reframing} introduced 
the \bt mechanism, developed mathematical models, and analyzed control 
performance with respect to frequency synchronization, the long-term 
behavior of buffer occupancy has remained largely unexplored. This 
paper fills this gap by providing a detailed analysis of the steady-state 
buffer occupancy dynamics in a \bt system. We utilize a fluid approximation 
of the abstract frame model~\cite{bms}, simplifying the analysis by ignoring 
quantization and sampling effects.  We prove that buffer occupancy 
converges to a steady-state value and derive an explicit formula for 
this value based on system parameters. This analysis enables
optimizing buffer sizes and minimizing system latency --- a key performance 
metric in distributed systems.  Our results directly relate network 
topology, physical latencies, and controller gains to steady-state buffer 
occupancy, offering valuable insights for designing and deploying 
efficient, logically synchronous datacenter-scale systems.

\section{Notation}
We have a directed graph $\mathcal{G} = (\mathcal{V}, \mathcal{E})$.
Number
the vertices and edges so that $\mathcal V = \{1, \dots,n\}$ and
$\mathcal E = \{ 1, \dots, m\}$.  Define the source incidence matrix $S \in\R^{n \times
  m}$ by
\[
S_{ie} = \begin{cases}
  1 & \text{if  node $i$ is the source of edge $e$} \\
  0 &\text{otherwise}
\end{cases}
\]
and the destination incidence matrix $D \in\R^{n \times m}$ by
\[
D_{ie} = \begin{cases}
  1 & \text{if node $i$ is the destination of edge $e$} \\
  0 &\text{otherwise}
\end{cases}
\]
The usual incidence matrix of the graph is then $B = S - D$.  Let
$\one$ be the vector of all ones, then $S^\tp \one = \one$,
$D^\tp \one = \one$, and 
$B^\tp \one = 0$. 
Let $d^\text{in}_i$ be the in-degree of node $i$
given by
\[
\din_i = \abs{
  \{
  j\in \mathcal V \mid j \to i 
  \}
  }
\]
A directed graph is called \eemph{strongly connected} or
\eemph{irreducible} if for every $i,j$ there exist directed paths
$i\to j$ and $j\to i$. Suppose $A$ is a nonnegative matrix whose
sparsity corresponds to the graph adjacency. Then irreducibility of
$A$ is defined by irreducibility of the graph. A matrix
$Q\in\R^{n\times n}$ is called \eemph{Metzler} if $Q_{ij} \geq 0$ for
all $i\neq j$, and it is called a rate matrix if in addition $Q\one =
0$. If $Q$ is Metzler and irreducible, then there is an eigenvalue
$\lambda^\text{metzler}$ which is real and which has positive left and
right eigenvectors. 

\section{Modeling}

\subsection{The \bt mechanism}

The bittide mechanism operates via the interplay of data transmission
and feedback control, eliminating the need for a global clock
reference, or for communication to exchange explicit time
measurements.

Each node possesses its own independent clock driven by an adjustable
oscillator, together with a number of network interfaces which directly connect
to other \bt nodes. A node sends frames to its neighbors, and by
observing the arrival rate of the data the receiving nodes deduce
information about the relative clock rate at the sender.  The content of these
frames is immaterial to the synchronization process; rather, it is the
act of transmitting and receiving frames that conveys the necessary
synchronization information.

Each node has a collection of FIFOs, called \emph{elastic buffers},
one for each incoming link. These specialized buffers serve as
temporary reservoirs for arriving data frames. With each clock cycle,
a node performs several actions. It removes the first frame from the
head of each elastic buffer, and stores it in memory for the processor
core to make use of. In addition, the processor core supplies a new
frame for each outgoing link. Therefore the rate at which frames are
sent is determined by the clock rate at the node, as is the rate at
which frames are removed from the elastic buffers. Conversely, at the
receiver, frames are added to the tail of the elastic buffer as they arrive.
The rate at which frames arrive is determined by the rate of
the sender's clock, adjusted for the physical latency of the link.

Since physical oscillators are inherently imprecise, there will be
slight variations in frequency between the nodes. Each elastic buffer
is drained at the local clock rate, but filled at the sender's clock
rate. Even a very slight difference in their frequencies will cause
the buffer to rapidly overflow or underflow.  To address this, the
bittide mechanism incorporates a feedback control system in which each
node monitors the occupancy levels of its elastic
buffers.  An increase in a buffer's occupancy signifies that the
node's clock is lagging behind that of the node transmitting data to
that buffer, while a decrease indicates the opposite.  Based on these
observations, each node dynamically adjusts its oscillator frequency,
ensuring that the buffer occupancies remain balanced.

\subsection{Abstract frame model}

We first consider a model of \bt called the \emph{abstract frame
model}, presented in~\cite{bms}. We develop mathematical analysis below
for a simplified fluid version of this model, without quantization.

Each node
$i\in\mathcal{V}$ has a clock with phase $\theta_i(t)$. The
\emph{localticks} of the clock are the times $t$ when $\theta_i(t)$ is
an integer. If $i\to j$ is a directed edge in the graph $\mathcal{G}$,
then with every localtick of the clock at node $i$ a frame is sent
from node $i$ to node $j$. We number the frames in transit by $k\in\Z$
equal to the localtick at the sender at the time of departure. The
frames in transit are those with $k$ satisfying
\[
\floor{\theta_j(t)} - \lambda_\itoj + 1
\leq k
\leq \floor{\theta_i(t)}
\]
We interpret this as follows. The left-hand inequality means that a
frame arrives with each localtick at node $j$, and the right-hand
inequality means that a frame is sent with each localtick at node $i$.
Suppose a frame is sent at localtick $\theta_i(t)=a$.  The clock at
the receiving node $\theta_j$ increases until $\theta_j(t) = a +
\lambda_\itoj$, at which time the frame is popped from the elastic
buffer at node $j$.

The above inequalities imply that the number of frames in transit
is
\[
\nu_\itoj(t) = \floor{\theta_i(t)} - \floor{\theta_j(t)} + \lambda_\itoj
\]
The frames move from node $i$ to node $j$ by first passing over a
physical link with latency $l_\itoj$ and then into the elastic buffer
at node $j$. The number of frames on the physical link is
\[
\gamma_\itoj(t) = \floor{\theta_i(t)} - \floor{\theta_i(t-l_\itoj)}
\]
and hence the occupancy $\beta_\itoj(t)$ of the elastic buffer is
\begin{align*}
  \beta_\itoj(t) &= \nu_\itoj(t) - \gamma_\itoj(t) \\
  &= \floor{\theta_i(t-l_\itoj)} - \floor{\theta_j(t)} + \lambda_\itoj
\end{align*}
The frequency of each oscillator is determined by a proportional
controller based on measurements of the buffer occupancies.
Specifically, the frequency correction at a node is set proportional
to the sum of the relative buffer occupancies for all incoming links.
Here, by \emph{relative}, we mean that we measure
the elastic buffer occupancy relative to the midpoint
of the buffer, so that at equilibrium the buffer will be half full,
and so that occupancies can be both positive and negative.
The oscillator at node $i$ has a base frequency $\wu_i$, which is
not known to the node. However, the oscillator is adjustable,
and the controller chooses the \emph{correction} $c_i$, so that
the oscillator frequency becomes $\wu_i + c_i$. 
This leads to the following model.
\begin{equation}
\begin{aligned}
  \label{eqn:cafm}
  \dot{\theta}_i(t) &= \wu_i +  c_i(t) \\
  c_i(t) & = \kp  \sum_{j \mid j \to i} (\beta_\jtoi(t) - \betaoff)        \\
  \beta_\itoj(t) &=  \floor{\theta_i(t-l_\itoj)} - \floor{\theta_j(t)} + \lambda_\itoj
\end{aligned}
\end{equation}
The main distinction between this model and a physical implementation
of \bt is that in this model the controller continuously updates the
frequency of the oscillators, whereas in a practical system the
controller is applied at a specific sample rate. The model with
sampling is called the abstract frame model, and is specified
in~\cite{bms}. In practice there is also a small amount of time
required for the controller updates to occur.  Both of these phenomena
can have a significant effect in certain parameter regimes, but in
this paper we do not analyze these effects.

\subsection{Fluid model}

\begin{defn}
  We define the list of system parameters $\mathcal{P} =
  (\theta^0,\wu,\lambda,k, \mathcal{G}, l, \betaoff)$. We call
  $\mathcal{P}$ admissible if
  \begin{itemize}
  \item the graph $\mathcal{G}$ is irreducible
  \item the gain is positive, \ie, $k > 0$
  \item the latencies are nonnegative, \ie, $l_e \geq 0$ for all $e\in\mathcal{E}$
  \item the unknown frequencies are positive, \ie, $\wu_i > 0$ for all $i\in\mathcal{V}$
  \end{itemize}
\end{defn}
Given an admissible parameter list $\mathcal{P}$ we will analyze the behavior
of the approximate model
\[
\begin{aligned}
  \dot{\theta}_i(t) &= \wu_i +  c_i(t) \\
  c_i(t) & = \kp  \sum_{j \mid j \to i} (\beta_\jtoi(t) - \betaoff)        \\
  \beta_\itoj(t) &=  \theta_i(t-l_\itoj) - \theta_j(t) + \lambda_\itoj
\end{aligned}
\]
which is a \emph{fluid approximation} of the original model~\eqref{eqn:cafm}.
This model ignores the effects of the granularity of the frames.

For convenience define the delay operator $\D$ which operates
on $\R^m$ valued signals by
\[
(\D z)_e(t) = z_e(t-l_e) \quad \text{for }e=1,\dots,m
\]
Then we may write the above model more concisely as
\begin{equation}
  \label{eqn:model}
  \begin{aligned}
    \dot\theta &= \wu + \kp  D(\beta - \betaoff)  \\
    \beta &= \D S^\tp \floor{\theta} - D^\tp \floor{\theta} + \lambda
  \end{aligned}
\end{equation}
This defines a delay differential equation for $\theta$, whose state at time $t$ is
given by $\theta$ restricted to the interval $[t-l^\text{max}, t]$
where $l^\text{max}$ is the maximum latency of any edge $l^\text{max}
= \max_e l_e$.
Since the forcing term in the delay differential equation is constant, we
make the following definition in terms of the system parameters.

\begin{defn}
  \label{def:v}
  We define the \eemph{input} $v\in\R^n$ by 
  \[
  v = \wu + kD(\lambda - \betaoff)
  \]
\end{defn}
We can also encapsulate the delay operator as
\[
  \A = kD(\D S^\tp - D^\tp)
\]
which allows us to express the dynamics~\eqref{eqn:model} for
$\theta$ in the particularly simple equivalent functional form
\begin{equation}
  \label{eqn:thetaaut}
  \dot\theta =  \A \theta + v
\end{equation}
Solutions to this equation always exist, which can be shown via the method
of steps.

\subsection{Limiting behavior}

Given the graph $\mathcal{G}$ we immediately have the incidence
matrices $D$ and $S$, and $B=S-D$. From these we may define a directed
Laplacian matrix $Q\in\R^{n\times n}$ by
\[
 Q = DB^\tp
\]
Note that $Q$ is a rate matrix, and specifically
\[
Q_{ij} = \begin{cases}
  -\din_i
  & \text{if } i = j \\
  1 & \text{if } j \to i  \in \mathcal{E} \\
  0 & \text{otherwise} 
\end{cases}
\]
Let $z$ be the left Metzler eigenvector of $Q$ which satisfies $z^\tp Q
= 0$, normalized so that $\one^\tp z = 1$ and $z>0$.
Let $L\in\R^{m\times m}$ be the diagonal matrix with $L_{ee} = l_e$, and
and define for convenience 
\[
X = I + kDLS^\tp \qquad H = \frac{\one z^\tp}{z^\tp X \one}
\]
Notice in particular that $HXH = H$, and $X$ is a generalized inverse of $H$.
We will use these in the following. 
\paragraph{Relative equilibrium.} 
For most choices of parameters $\mathcal P$ the
system~\eqref{eqn:model} does not have an equilibrium, but it does
have a type of relative equilibrium, as the following result shows.
\begin{thm}
  \label{thm:eqm}
  Suppose the parameters $\mathcal P$ are admissible.
  Then 
  there exists $\wss,\thetass\in\R^n$ such that
  \[
  \theta(t) = \wss t   + \thetass
  \]
  satisfies~\eqref{eqn:model}. Such $\wss$ is unique, given by
  \[
  \wss = H v
  \]
  and $\thetass$ is not unique in general but must satisfy
  \[
 \qquad kQ \thetass =   (X H - I )v
  \]
\end{thm}
\begin{proof}
  Such a function is a solution if and only if
  \begin{equation}
    \label{eqn:releq}
    \bmat{Q & 0 \\ X & -kQ  }
    \bmat{ \wss \\ \thetass}
    = \bmat{0 \\ v}
  \end{equation}
  Here we have used the property $\D S^\tp t x = (tI-L)S^\tp x$
  for any $x\in\R^n$. Since the graph is irreducible we have $Q \wss = 0$
  iff $\wss = \bwss\one$ for some scalar $\bwss$. Then~\eqref{eqn:releq}
  holds iff
  \begin{equation}
    \label{eqn:rank}
    v = M \bmat{\bwss \\ \thetass}
  \end{equation}
  where $M = \bmat{X\one  & -kQ}$. We now show
  that $\rank(M) = n$, which implies that $\bwss, \thetass$
  satisfying~\eqref{eqn:rank} must exist.  Let $z$
  be the positive left Metzler eigenvector of $Q$.
  Then for any nonzero nonnegative $x\in\R^n$, we
  have $z^\tp x > 0$, but since $z^\tp Q =0$ such $x$ cannot lie in
  the range of $Q$.  Since $k>0$ and $l \geq 0$, we therefore have
  $X\one  \not\in\range(Q)$. Since $Q$ is irreducible we have
  $\rank(Q) = n-1$ and hence $\rank(M) =n$.

  To show uniqueness, we have
  \begin{align*}
    z^\tp v &= z^\tp M \bmat{\bwss \\ \thetass}\\
    &= \bwss z^\tp X \one 
  \end{align*}
  which gives the desired expressions for $\bwss$ and $\thetass$.
\end{proof}

\paragraph{Relative dynamics.} We can now construct the relative dynamics. Suppose $\wss$ and
$\thetass$ satisfy the conditions of Theorem~\ref{thm:eqm}. Let
\[
\zeta(t) = \theta(t) - \wss t \one  - \thetass
\]
Then we have the autonomous functional differential equation
\begin{equation}
  \label{eqn:homog}
  \dot \zeta = \A \zeta
\end{equation}
We can also write this as
\[
\dot{\zeta}_i(t) = \kp  \sum_{j \mid j \to i}
\bl(
\zeta_j(t -  l_\itoj) - \zeta_i(t) 
\br)
\]
This is the well-studied \emph{consensus dynamics}.

\paragraph{Stability.} For the relative dynamics, 
the following result is known.
\begin{thm}
  Suppose the parameters $\mathcal P$ are admissible.
  For any continuous initial
  conditions $\phi_0:[-l^\text{max},0] \to \R^n$ there exists a continuous function
  $\zeta:[-l^\text{max},\infty) \to \R^n$ which
    satisfies~\eqref{eqn:homog} with initial conditions $\zeta(t)=\phi(t)$ for $t\leq 0$
    and a constant $\bzetass\in\R$ such that 
  \[
  \lim_{t \to \infty} \zeta(t) = \zetass
  \]
  where $\zetass = \bzetass\one$. 
  This convergence is exponential, and $\lim_{t \to \infty} \dot\zeta = 0$.
\end{thm}
For the asymptotic stability part of this result, a direct proof is
given in~\cite{schmidt2009} using Lyapunov-Razumikhin
functions.  Bounds on the
exponential rates are given in~\cite{somarakis2014}.  As discussed in
Section~\ref{sec:relatedwork}, there are several approaches to proving
variants of this result. Using this, we immediately have the following
corollary.
\begin{cor}
  \label{cor:wss}
  Suppose the parameters $\mathcal P$ are admissible,
  and $\theta$ satisfies the dynamics~\eqref{eqn:model}.
  Then for any initial conditions we have
  \[
  \lim_{t \to \infty} \dot\theta(t) = H v
  \]
\end{cor}

\begin{proof}
  Let $\wss = Hv$ and $\thetass$ be as in Theorem~\ref{thm:eqm}. Now
  let $\zeta(t) = \theta(t) - \wss t   - \thetass$ and observe that $\zeta$
  satisfies the relative dynamics $\dot\zeta = \A \zeta$. Since $\dot\zeta \to 0$, we have
  \[
  \lim_{t\to\infty} \dot\theta(t) = \wss
  \]
  as desired.
\end{proof}

\section{Limiting buffer occupancy}

We can now show that the buffer occupancy converges, and calculate its
steady-state value. Let $Q^\hash$ be a generalized inverse of $Q$,
and define $Y = B^\tp Q^\hash D - I$.
\begin{thm}
  \label{thm:bss}
  Suppose the parameters $\mathcal P$ are admissible, and $\theta$ and
  $\beta$ satisfy~\eqref{eqn:model}.  Then $\beta(t)$ converges
  as $t  \to \infty$.  The limit is given by
  $\lim_{t\to\infty} \beta(t) =  \betass$ where
  \[
  \betass = \lambda + Y L S^\tp H v  + \kp^{-1}  B^\tp Q^\hash (H - I) v
  \]
\end{thm}
\begin{proof}
  As in Corollary~\ref{cor:wss} we have, as $t \to \infty$, 
  \[
   \theta(t) - \wss t - \thetass \to  \zetass
  \]
  hence
  \[
  (\D S^\tp - D^\tp) ( \theta(t) - \wss t - \thetass - \zetass) \to 0
  \]
  Now  using the fact that $B^\tp \wss = 0$ and $B^\tp \zetass = 0$ and the
  property that $\D S^\tp t x = (tI-L)S^\tp x$ for any $x\in\R^n$ results in
  \[
   (\D S^\tp - D^\tp)  \theta(t) \to  B^\tp\thetass  - LS^\tp \wss
  \]
  Therefore, with
  \[
  \betass = \lambda + B^\tp \thetass - L S^\tp \wss 
  \]
  we have $\beta(t) \to \betass$ as $t \to \infty$. We simplify this as follows. From
  Theorem~\ref{thm:eqm} we have $\kp Q\thetass = (XH-I)v$ and $\wss = Hv$. Therefore
  \[
  \kp\thetass = Q^\hash(XH - I) v + y
  \]
  for some $y\in\Null(Q)$.
  Since $\mathcal{G}$ is irreducible
  we have
  $\Null(Q) = \Span\{\one\}$ and so $B^\tp y =0$, therefore
  \[
  \betass = \lambda + \kp^{-1} B^\tp Q^\hash(XH-I)v  - L S^\tp Hv
  \]
  Now using $X= I  + \kp D L S^\tp$   the result follows.
\end{proof}

\paragraph{The zero-latency case.} Earlier work~\cite{reframing} has considered
models of \bt with the approximation that the latency is very small or
zero. We recall here the results in that case, and show that they are
consistent with the above formulae. First, we consider the limiting
frequency. When $l = 0$ the matrix $H$ becomes $H=\one z^\tp$ and so
Corollary~\ref{cor:wss} gives
\[
\wss = \lim_{t \to \infty} \dot\theta(t) = \one z^\tp v
\]
The matrix $H$ in this case is a projection. For the steady-state buffer occupancy,
Theorem~\ref{thm:bss}
reduces to
\[
\betass = \lambda + \kp^{-1} B^\tp Q^\hash(\one z^\tp - I)v
\]
Let the eigendecomposition of $Q$ be $QT = TD$ where
\[
D = \bmat{0 & 0 \\ 0 & \Lambda}
\quad
T = \bmat{\one & T_2}
\quad
T^{-1} = \bmat{z^\tp \\ V_2^\tp}
\]
for appropriate matrices $T_2$, $V_2$ and $\Lambda$. Define
\[
Q^S = T \bmat{0 & 0 \\ 0 &\Lambda^{-1}} T^{-1}
\]
Then if we choose $Q^\hash = Q^S$ we have $Q^\hash \one = 0$ and 
\[
\betass = \lambda - \kp^{-1} B^\tp Q^S v
\]
which matches the case with zero latency which was  analyzed in Lemma~5 of~\cite{reframing}.

\subsection{Steady initial conditions}  
A simple set of initial conditions for the system is as follows.  We
make the assumption that all of the controllers are turned on
simultaneously. This is an idealization, but often in practice the
controllers start up close together in time; for example, when all
machines are in a rack and power is applied to the rack.  In that
case, the oscillators start at their base frequencies $\wu$, and the
controllers start to correct frequency as soon as all of the links
achieve frequency lock at the deserializer modules in the nodes. Then
we model this by assuming that the that frequencies of the oscillators
$\dot\theta$ are equal to $\wu$ for the period of time
$[-l^\text{max}, 0)$. Another important property at startup is
that the buffer offsets are set equal to the initial occupancy;
that is, we set $\betaoff = \beta(0)$.

\begin{defn}
  We say that the system parameters $\mathcal P$
  satisfy \eemph{steady initial conditions} if
  \begin{equation}
    \label{eqn:sic}
    \betaoff = \lambda - LS^\tp \wu + B^\tp \theta^0
  \end{equation}
\end{defn}
We make this definition because the buffer occupancy satisfies
\[
\beta = \lambda + \D S^\tp \theta - D^\tp \theta
\]
Under steady initial conditions $\beta(0) = \betaoff$ and
for $t\leq 0$ we have $\theta(t)= t\wu  + \theta^0$, and these
two conditions together imply  that equation~\eqref{eqn:sic} holds.

\begin{lem}
  \label{lem:vsteady}
  Suppose the system starts in steady initial conditions.  Then
  \[
  v = X \wu - kQ \theta^0 
  \]
\end{lem}
\begin{proof}
  From Definition~\ref{def:v} we have
  \begin{align*}
    v &= \wu + kD(\lambda - \betaoff) \\
    &= \wu + kD(LS^\tp \wu - B^\tp \theta^0) \\
    &= X \wu - k Q \theta^0
  \end{align*}
  as desired.
\end{proof}

We now turn to the steady-state frequency under these initial conditions.
\begin{lem}
  Suppose that the system parameters $\mathcal P$ satisfy steady initial conditions.
  Then the steady-state frequency is
  \[
  \wss = H X \wu
  \]
\end{lem}
\begin{proof}
  From Theorem~\ref{thm:eqm} and Corollary~\ref{cor:wss} we have
  $\wss = Hv$. Now using Lemma~\ref{lem:vsteady} and
  the fact that $HQ = 0$ gives the desired result.
\end{proof}
An important observation here is that the matrix $HX$ is a
projection. The interpretation of this is as follows. Since $\range(H)
= \Span\{\one\}$ this means that if all of the unknown frequencies are
the same, then the steady-state frequency is equal to the unknown base
frequencies. This is as expected, since with all frequencies equal,
and with steady initial conditions, we will have all buffer occupancies
remain at the buffer offset point.

We now turn to the steady-state buffer occupancy.

\begin{thm}
  \label{thm:bsteady}
  Suppose that the system parameters $\mathcal P$ satisfy steady initial conditions.
  Then the steady-state buffer occupancy is
  \begin{multline}
    \label{eqn:bssc}
    \betass - \betaoff = (LS^\tp + YLS^\tp HX) \wu \\
    + k^{-1} B^\tp Q^\hash (H-I) X \wu
  \end{multline}
\end{thm}
\begin{proof}
  From Theorem~\ref{thm:bss}, using  $v = X \wu - kQ \theta^0 $, we have
  \begin{align*}
    \betass
    &= \lambda + Y L S^\tp H v  + \kp^{-1}  B^\tp Q^\hash (H - I) v\\
    &= \lambda + Y L S^\tp H X \wu +  \kp^{-1}  B^\tp Q^\hash H X \wu
    - \kp^{-1}  B^\tp Q^\hash v
  \end{align*}
  The last term in this expression is
  \begin{align*}
    \kp^{-1}  B^\tp Q^\hash v &= \kp^{-1}  B^\tp Q^\hash X \wu - B^\tp Q^\hash Q \theta^0
    \\
    & = \kp^{-1}  B^\tp Q^\hash X \wu - B^\tp \theta^0
  \end{align*}
  where we have used the fact that $B^\tp Q^\hash Q = B^\tp$ since $\Null(Q) = \Null(B)$.
  Now we use the steady initial conditions property that
  \[
  \betaoff = \lambda - LS^\tp \wu + B^\tp \theta^0
  \]
  to conclude the desired result.
\end{proof}

This result is helpful in the design of \bt systems since it allows
determination of the limiting buffer occupancies given the physical
latencies $L$ and graph properties $Y,B,S,Q$.  Notice that the
relative buffer occupancy in~\eqref{eqn:bssc} is the sum of two
terms. The first term tends to zero as the latency tends to zero.

We perform a simple simulation using the Callisto
simulator~\cite{callisto} to illustrate this result. Callisto
simulates the abstract frame model, including both non-uniform
sampling and quantization, whereas the analysis here is for the linear
model. Nonetheless, the correspondence between the mathematical
predictions and the simulation are very close. Figure~\ref{fig:graph}
shows the graph which we use. All links have physical latency $l =
2.7\times 10^{-7}$, the gain is $k=0.25$, initial $\theta^0=0.1$, and
all links have logical latency $\lambda = 34$. In particular, these
parameters correspond to those of the open-source \bt hardware
at~\cite{qbay}. Figure~\ref{fig:freq} shows the frequency behavior,
including a black line showing the predicted
$\wss$. Figure~\ref{fig:mocc} shows the relative buffer occupancies
include black lines at the predicted steady-state values
from Theorem~\ref{thm:bsteady}.

\begin{figure}[ht!]
  \centerline{\includegraphics[width=0.6\linewidth]{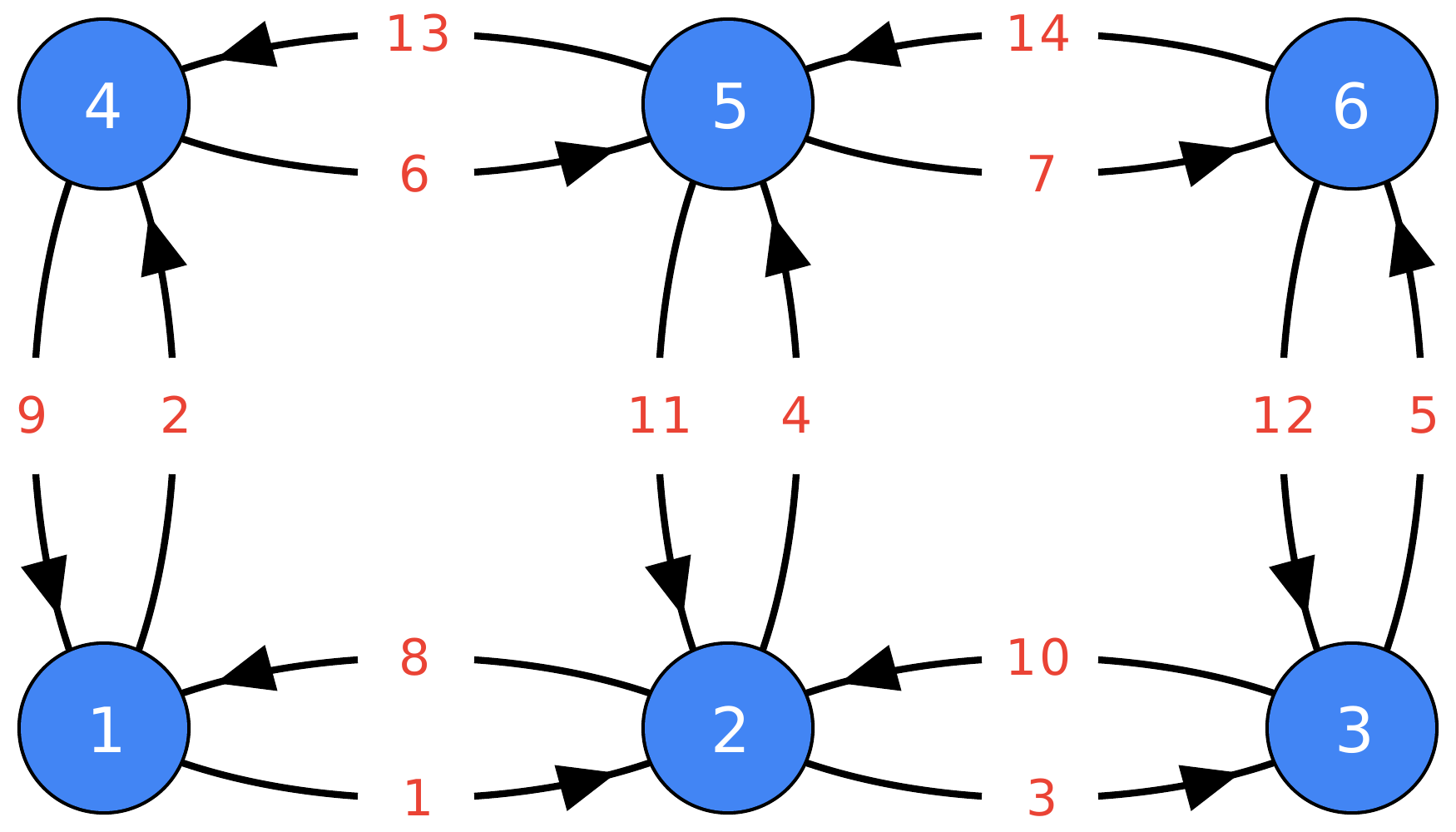}}
  \caption{Graph showing edge and node numbering}
  \label{fig:graph}
\end{figure}

\def\ts#1{\fboxsep=2pt\colorbox{white}{\small #1}}
\begin{figure}[ht!]
  \centerline{\begin{overpic}[width=0.9\linewidth]{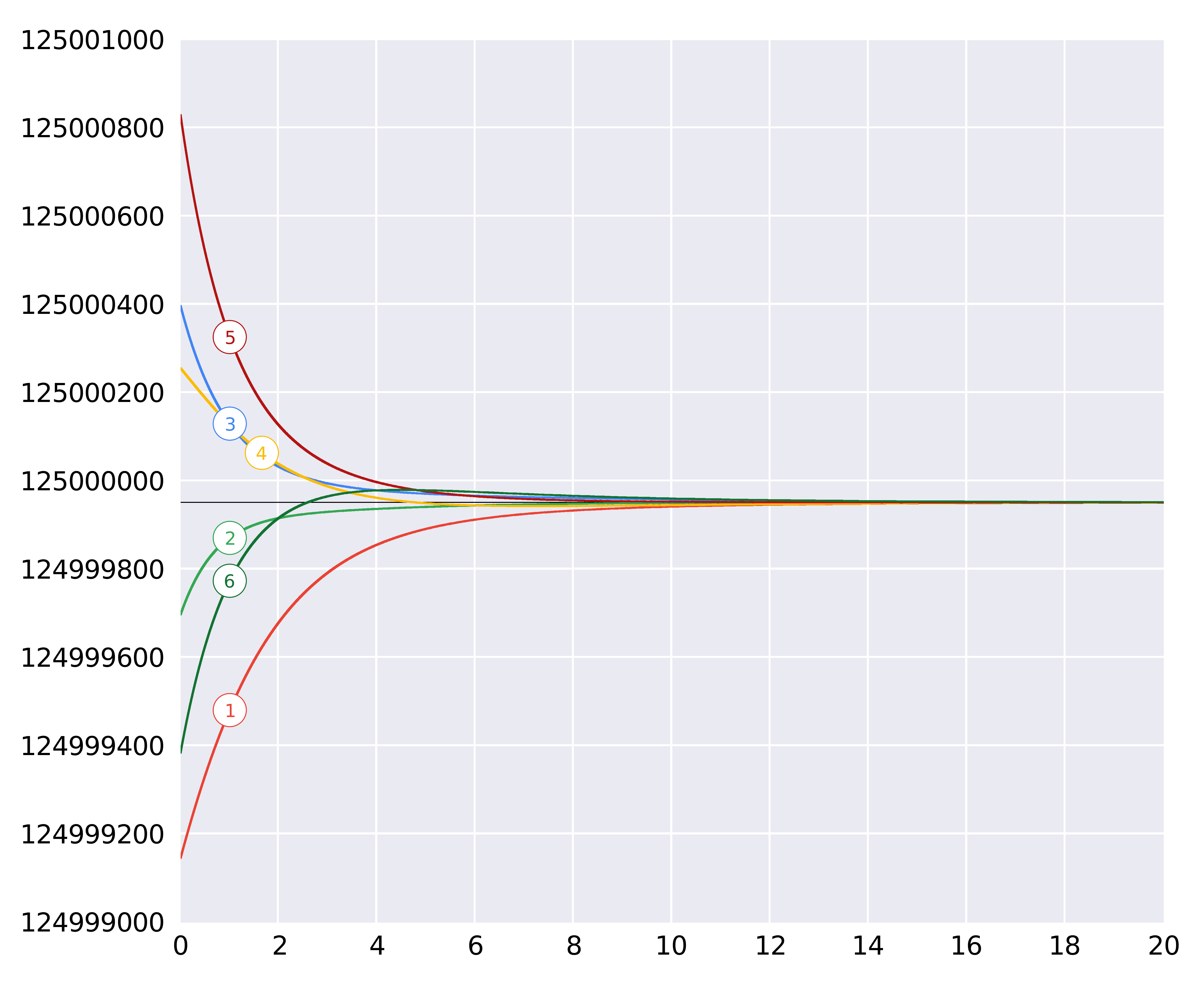}
      \put(55,0){\small $t$}
      \put(-3,40){\small $\dot\theta$}
    \end{overpic}}
    \caption{Per-node frequencies}
  \label{fig:freq}
\end{figure}

\begin{figure}[ht!]
  \centerline{\begin{overpic}[width=0.9\linewidth]{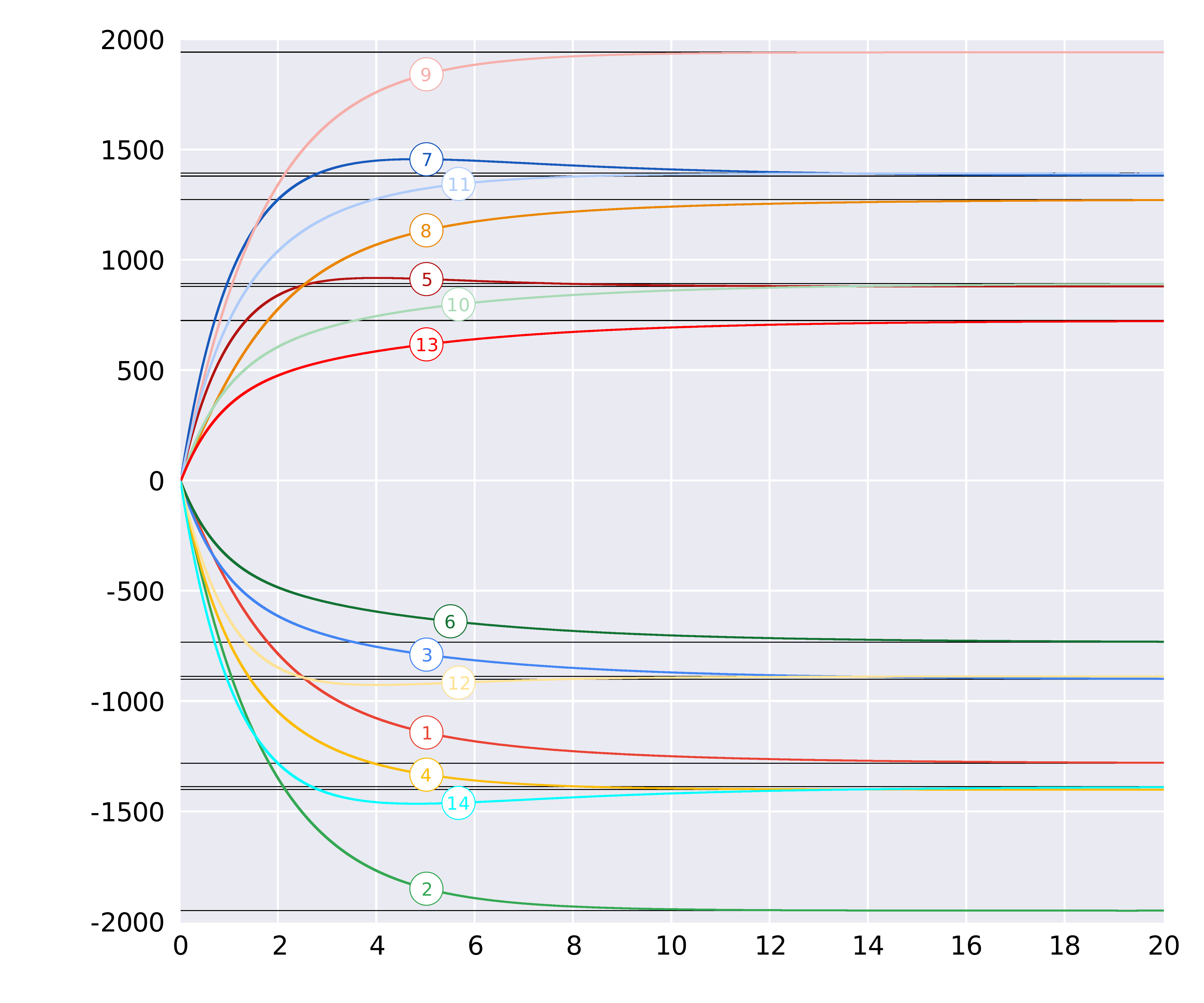}
      \put(55,0){\small $t$}
      \put(10,43){\llap{\small $\dot\beta-\betaoff$}}
    \end{overpic}}
    \caption{Per-edge buffer occupancies}
  \label{fig:mocc}
\end{figure}

\section{Related work}

\label{sec:relatedwork}

The challenge of achieving synchrony in distributed systems has been a
focus of research for decades. Traditional approaches often rely on
aligning local clocks to a global time reference, such as Coordinated
Universal Time (UTC), using protocols like Precision Time Protocol
(PTP)~\cite{ptp}.  However, these methods can suffer from limitations
in accuracy and scalability, especially in large-scale data center
environments.  The \bt approach is fundamentally different; instead of
synchronizing the absolute time of the clocks, \bt instead achieves
logical synchrony via controlling their average frequency.

This concept was first introduced in~\cite{spalink_2006} and further
developed in~\cite{bms,ls,res,reframing}, demonstrating its potential
for precise coordination without the overhead of traditional
synchronization protocols.  The abstract frame model is a
cycle-accurate model of \bt, developed in~\cite{bms}. A linear
approximate model of \bt was also developed, in which the
frequency dynamics are related to the well-known consensus
dynamics~\cite{olfati2004,jadbabaie,boyd2004}. A linear model for the
\bt mechanism adds two key components to consensus dynamics.  The
first is a forcing term, due to the continual increase of the clocks.
The presence of a forcing term, in combination with delays, has been
analyzed in~\cite{schmidt2009}. The second addition is buffer
occupancies~\cite{bms}, a differential output of central importance in
a practical \bt system.

In this paper, we make use of existing stability results for linear
consensus dynamics with delays. The stability of consensus dynamics
has been investigated in many papers, using a plethora of mathematical
techniques.  With these diverse
approaches comes several different sets of technical
assumptions on the graphs and the latencies. 
 Stability of linear delay systems is substantially more
complicated than for finite-dimensional systems. Several works make
use of the Nyquist theorem and variants thereof, in
particular~\cite{lee2006,munz2010nyquist,lestas2007}.  Lyapunov
methods and Razumikhin theory have also been used to prove stability of
consensus models with heterogeneous latencies
in~\cite{munz2008,munz2010,haddad,somarakis2015,moreau2004}.  A
different approach using fixed-point theory is used
in~\cite{somarakis2014}. Discrete-time models are considered
in~\cite{blondel2005,tsitsiklis1986}. 

\section{Conclusions}

This paper analyzed the steady-state buffer occupancy 
in \bt systems, a crucial factor determining overall 
system latency. By using a fluid approximation of the 
\bt model, we proved that buffer occupancy converges 
to a steady-state value and derived an explicit formula 
for this value as a function of system parameters. This 
formula incorporates network topology, physical latencies, 
and controller gains, providing valuable insights for 
system designers.  Understanding the long-term behavior of 
buffer occupancy enables optimization of buffer sizes, 
minimizing latency, and ensuring efficient resource 
utilization in \bt-based distributed systems.

Our analysis, however, relies on certain simplifying 
assumptions.  Future work should address the limitations of 
the fluid model by incorporating the effects of quantization 
due to discrete frame transmission.  Further investigation is 
also needed to analyze the impact of discrete-time control 
implementation, where controllers operate at specific sampling 
rates, as opposed to the continuous control assumed in this 
paper.  Finally, extending the analysis to include dynamic changes 
in network topology and node failures would enhance the practical 
applicability of these results.  Addressing these open questions 
will pave the way for a more comprehensive understanding of \bt 
system dynamics and enable the design of robust and efficient 
large-scale logically synchronous distributed systems.

\section{Acknowledgments}

The authors would like to thank C\u{a}lin Ca\c{s}caval, Pouya
Dormiani, and Martin Izzard for their insightful discussions and
invaluable feedback throughout this research project.

\bibliographystyle{abbrv}

\end{document}